\documentclass[11pt]{article}
\usepackage[a4paper, margin=1in]{geometry}
\usepackage{microtype}
\usepackage{graphicx}
\usepackage{subcaption}
\usepackage{hyperref}
\usepackage{amsmath}
\usepackage{amssymb}
\usepackage{mathtools}
\usepackage{amsthm}
\usepackage{parskip}

\makeatletter
\def\thm@space@setup{\thm@preskip=\parskip \thm@postskip=0pt}
\makeatother

\usepackage{times}
\usepackage[capitalize,noabbrev]{cleveref}
\usepackage[dvipsnames]{xcolor}
\usepackage{dsfont}
\usepackage{enumitem}
\usepackage{tikz}
\usetikzlibrary{positioning}
\usepackage{natbib}
\usepackage{makecell}
\usepackage{colortbl}
\usepackage{float}
\usepackage{xurl}
\usepackage{caption}

\theoremstyle{plain}
\newtheorem{theorem}{Theorem}[section]

\newtheorem{lemma}[theorem]{Lemma}
\newtheorem{observation}[theorem]{Observation}

\theoremstyle{definition}

\theoremstyle{remark}

\DeclareMathOperator*{\argmin}{arg\,min}
\DeclareSymbolFontAlphabet{\mathbb}{AMSb}

\newcommand{\disjoint}{\text{DISJ}}
\newcommand{\defn}{\emph}
\newcommand{\grayhl}[1]{\colorbox{gray!20}{$#1$}}

\title{\textbf{Stable Matching with Predictions: \\ Robustness and Efficiency under Pruned Preferences}}
\author{
Samuel McCauley\\ Williams College\\ \texttt{sam@cs.williams.edu}
\and 
Benjamin Moseley\thanks{B. Moseley and H. Niaparast were supported in part by Google Research Award, an Infor Research Award, and NSF grant CCF-2121744.}\\ Carnegie Mellon University\\ \texttt{moseleyb@andrew.cmu.edu} 
\and 
Helia Niaparast\footnotemark[1]\\ Carnegie Mellon University\\ \texttt{helia@cmu.edu}
\and
Shikha Singh\\ Williams College\\ \texttt{shikha@cs.williams.edu}
}
\date{}

\begin{document}

\maketitle

\begin{abstract}
In this paper, we study the fundamental problem of finding a stable matching in two-sided matching markets.  In the classic variant, it is assumed that both sides of the market submit a ranked list of all agents on the other side.  However, in large matching markets such as the National Resident Matching Program (NRMP), it is infeasible for hospitals to interview or mutually rank each resident. In this paper, we study the stable matching problem with truncated preference lists.  In particular, we assume that, based on historical datasets, each hospital has a predicted rank of its likely match and only ranks residents within a bounded interval around that prediction. 

We use the \emph{algorithms-with-predictions} framework and show that the classic deferred-acceptance (DA) algorithm used to compute stable matchings is robust to such truncation. We present two algorithms and theoretically and empirically evaluate their performance.  Our results show that even with reasonably accurate predictions, it is possible to significantly cut down on both \emph{instance size} (the length of preference lists) as well as the number of proposals made.
These results explain the practical success of the DA algorithm and connect market design to the emerging theory of algorithms with predictions.
\end{abstract}

\section{Introduction}

Centralized matching markets such as the National Resident Matching Program (NRMP) are among the most visible successes of market design. Each year, tens of thousands of residents and hospitals submit ranked preference lists, and a centralized algorithm assigns residents to hospitals. In 2024, the NRMP included over 50,000 applicants competing for roughly 40,000 residency positions~\citep{nrmp2024match}.

The theoretical foundation underlying these two-sided markets is the \defn{stable matching problem}~\citep{gale1962college}. In the classical problem, the agents on the two sides---the residents and hospitals---submit a ranked list of preferences over the opposite side. The goal of the problem is to find a \emph{stable} matching: an outcome where no pair of agents would both prefer each other over their assigned match (that is, there are no ``blocking pairs''). This guarantee of stability is the primary reason for the success and longevity of matching markets like the NRMP~\citep{roth1984evolution}.

While the classic stable matching problem assumes that participants submit complete preference lists over all potential partners, this is  infeasible in practice. Hospitals must decide whom to interview given limited time and resources. Consequently, they rely on past experience to form expectations about where they will match, adopting a data-driven approach to interviewing.

For example, an elite hospital may interview highly competitive candidates until it is confident it has reached its eventual match. Meanwhile, a mid-tier hospital may skip candidates who are unlikely to accept, focusing instead on a contiguous \emph{window} of mid-ranked applicants who are similar to those who matched there in the past. This behavior reflects a natural adaptation to capacity constraints: each hospital implicitly \emph{predicts} the approximate rank of its eventual match and restricts attention to residents near that prediction.

This predictive truncation highlights a difference between how such markets are studied in theory versus how they operate in practice.  We study two fundamental questions to bridge this gap between market design theory and practice:

\begin{enumerate}
    \item \textbf{Market Dynamics (Descriptive):} Does this heuristic behavior threaten the stability of the market? If agents only propose to a ``predicted'' set of candidates, do we still reach a stable matching?
    \item \textbf{Algorithmic Efficiency (Prescriptive):} Can we exploit these predictions to design faster algorithms? Specifically, can we bypass the quadratic runtime of standard stable matching algorithms by leveraging predictions of where hospitals will match?
\end{enumerate}

\subsection{Our Model: Algorithms with Predictions}

We study the stable matching problem with pruned preference lists in the \emph{algorithms with predictions} framework~\citep{mitzenmacher2022algorithms,lykouris2021competitive}. This new framework, also called \emph{learning-augmented algorithms}, is motivated by the premise that real-world applications repeatedly solve a problem on similar instances that share a common
underlying structure. Algorithms designed in this framework leverage predictions about these input instances to optimize efficiency on future computations.  This model has been extremely successful in improving algorithmic performance; see Section~\ref{sec:related}. 

In this model, the algorithm is given a prediction about the input instance and its performance is measured as a function of the prediction quality (as well as the input size).  In particular, the performance of the algorithm must (a) improve upon standard solutions if the predictions are good, (b) degrade proportionally to the prediction error, and (c) be robust to errors in predictions.

\subsection{Stable Matching with Predictions}
We study the stable matching problem in the algorithms-with-predictions model. We first review the classic problem.

Consider a two-sided matching market with $n$ residents $R$ and $n$ hospitals $H$.  Each agent $a$ (resident or hospital) has a preference list $L(a)$ ranking all agents on the other side.  We denote the stable matching instance by $I = (R \sqcup H, L)$. For a matching $\mu$ on $I$, let $\mu(a)$ denote $a$'s match in $\mu$.  A matching $\mu$ is \emph{perfect} if all agents are matched and is stable if it contains no blocking pairs.  A pair $(r,h)$ is a \emph{blocking pair} if $r$ prefers  $h$ to $\mu(r)$ and $h$ prefers $r$ to $\mu(h)$. 

The classic stable matching algorithm is the deferred-acceptance (DA) algorithm by \citet{gale1962college}.  The algorithm designates one side of the market as the proposing side and the other as the receiving side. In each round, every unmatched proposing agent makes an offer to the highest-ranked agent on its preference list who has not yet rejected it. Each agent on the receiving side then tentatively accepts its most preferred offer and rejects the others. The process continues until there are no rejections and all tentative matches become final. 

It is well-known that the DA algorithm produces a stable matching that is optimal for the proposing side and pessimal for the receiving side; see Theorem~\ref{thm:DA-optimality}.

\paragraph{Prediction model and truncation algorithms.}
Given a stable matching instance $I$, let $\mu$ be any stable matching on $I$. We assume that each hospital receives a predicted rank of the resident it matches with in $\mu$. That is, each hospital $h_i$, receives a \textbf{prediction $\rho_i$} of the rank of $\mu(h_i)$ on its preference list $L(h_i)$.  
Each hospital can then prune its complete preference list $L(h_i)$ using $\rho_i$ and submit this pruned list to the stable matching algorithm.

We propose two variants of the classic deferred-acceptance algorithm. 
In the first variant, the \textbf{window-truncated deferred-acceptance algorithm (WDA)},  each hospital also receives a prediction error bound $\eta_i$.  This estimate on the error is essentially a promise that its match $\mu(h_i)$ must have rank between $\rho_i-\eta_i$ and  $\rho_i+\eta_i$ on its list.  Thus, the hospitals can prune all other elements, leaving the contiguous sublist $\hat{L}(h_i) := L(h_i)[\rho_i-\eta_i~:~\rho_i+\eta_i]$.  In the second variant, the \textbf{prefix-truncated deferred-acceptance algorithm (PDA)}, the prediction error $\eta_i$ is unknown a priori and the hospitals prune their lists to retain the top $\rho_i$ ranks, that is, the sublist $\hat{L}(h_i) := L(h_i)[1~:~\rho_i]$. In both variants, each resident $r_j$'s list is pruned to remove any hospitals $h_i$ such that $r_j \notin \hat{L}(h_i)$; retaining relative ordering.  We use $\hat{I}$ to denote the resulting pruned stable-matching instance.

\subsection{Our Contributions}

\paragraph{Window-truncated DA.} Our analysis of the window-truncated DA algorithm reveals that robustness of stability depends crucially on which side of the market proposes.  Since, for every hospital $h_i$, the match $\mu(h_i)$ is promised to exist in the hospital's pruned list $\hat{L}(h_i)$, it is tempting to assume that the hospital-proposing DA on the pruned instance $\hat{I}$ is sufficient to output a correct stable matching.  
As our first result, we show this is not the case. In particular, Lemma~\ref{lem:hospital-proposing-is-unstable} shows that if the \emph{hospital-proposing} DA is run on $\hat{I}$, the resulting matching may be \emph{unstable} on the original instance $I$---even if $\mu$ is the output of the hospital-proposing DA on $I$. 
This highlights a market-design vulnerability: the hospital-proposing mechanism can be destabilized by data-driven truncation. This is particularly relevant given the history of the NRMP, which used a hospital-proposing DA in its early years~\citep{roth1984evolution}.

On the other hand, in Theorem \ref{thm:WDA}, we show that the resident-proposing DA, when run on the pruned instance $\hat{I}$, does produce a stable matching.      This provides a new theoretical justification for the NRMP’s design: it remains stable under bounded prediction errors and natural strategic pruning by hospitals. Moreover, we show that the performance of the resident-proposing WDA algorithm is proportional to the prediction error, i.e., it makes $O(\sum_i \eta_i)$ proposals.  Thus, if the prediction quality is good, the WDA algorithm has two major advantages:  it reduces the size of preference lists and the number of proposals made by DA significantly.

\paragraph{Prefix-truncated DA.}  When bounds on the prediction error $\eta_i$ are not available, we propose the prefix-truncated DA algorithm as an alternative.  This truncation strategy captures the realistic setting where hospitals only rank their top $\rho_i$ residents.  We show that running DA on the prefix-truncated instance $\hat{I}$ provides a certificate of which hospitals' predictions were inaccurate.   In particular, Theorem~\ref{thm:unmatched-subseteq-underpredicted} shows that the unmatched hospitals in the matching output by the PDA algorithm are a subset of those whose predictions are inaccurate.  This enables the algorithm to adapt to such erroneous predictions by selectively extending the lists of unmatched hospitals and rerunning until a perfect matching is obtained.  Moreover, when a stable matching is promised to exist within the pruned lists $\hat{L}$, in Theorem~\ref{thm:pda_running_time} we bound the number of proposals made by the PDA algorithm by the distance from the predicted matching to the nearest stable matching. 

\paragraph{Lower bound on efficiency if prediction-based truncation eliminates all stable matchings.}  Our theoretical guarantees on the efficiency of both the WDA and PDA algorithms rely on the assumption that, while the predictions can be quantitatively inaccurate, \emph{at least one stable matching} on the original instance $I$ is promised to exist within the pruned instance $\hat{I}$.  We prove a lower bound that shows that such a promise is in fact necessary for efficiency.  In particular, Theorem~\ref{thm:lower} shows that even with ``near-perfect'' predictions, without this promise, finding a stable matching takes $\Omega(n^2)$ time in the worst case, matching the worst-case performance of the standard DA algorithm.  This lower bound highlights an efficiency bottleneck that is fundamental to the stable matching problem---verifying if a prediction retains a stable solution is just as hard as solving the problem without any predictions.  

\paragraph{Empirical evaluation.}  Finally,  we validate our theoretical results through simulations on several market distributions.  
Our experiments evaluate how WDA and PDA perform on these distributions using empirically learned predictions.  In particular, our experimental results do not assume that the pruned lists generated from the learned predictions retain a stable matching. We observe that even without such a promise, the learned predictions tend to be quite robust.  In particular, the WDA algorithm recovers a stable matching most of the time, and both the WDA and PDA algorithms cut down on the number of proposals significantly over the standard DA algorithm.  Moreover, while the PDA algorithm is guaranteed to always output a stable matching, its instance size (preference list lengths) is a lot larger than the WDA algorithm.  Thus, the WDA algorithm can be an effective heuristic if the market admits very short preference lists.

\paragraph{Recovering intermediate stable matchings.}
The DA algorithm is well-known to be inequitable as it outputs the proposer-optimal and receiver-pessimal stable matching. Even if a given instance may have many intermediate and fairer stable matchings, finding them in an efficient way is an active research area; see e.g.~\citet{dworczak2021deferred,tziavelis2019equitable}.  In particular, \citet{dworczak2021deferred} give the first algorithm that is capable of outputting any intermediate stable matching.  However, their algorithm runs in $O(n^4)$ time.  \citet{tziavelis2019equitable} show that it is possible to output some intermediate matchings in $O(n^2)$ time.  Our paper complements these results by showing that it is possible for DA to output intermediate matchings efficiently by truncating the lists using predictions that include them.

\subsection{Related Work}\label{sec:related}

\paragraph{Stable matchings.} Stable matching  has inspired a vast body of research, including studies on the lattice structure of stable matchings~\citep{blair1984every, blair1988lattice}, extensions with ties and incomplete lists~\citep{irving1994stable, kavitha2007strongly}, and dynamic variants~\citep{bampis2023online, blum2002timing, bredereck2020adapting}. See surveys~\citet{hosseini2024strategic,DBLP:journals/corr/abs-1904-08196} for a comprehensive overview.

Beyond theory, stable matching has significant real-world applications. The NRMP is the most prominent example, but similar resident–hospital matching systems operate in other countries, including CaRMS in Canada and JRMP in Japan. Variants of the model have also been applied to school choice systems \citep{abdulkadirouglu2005new, abdulkadirouglu2005boston, correa2019school}, kidney exchange programs \citep{roth2005pairwise, irving2007cycle}, and labor markets \citep{kelso1982job}. These applications illustrate the  impact of stable matching theory on practical marketplaces.

Stable matchings with incomplete or {partial} preferences have been studied in the context of \emph{random matching markets}, where the preferences of the agents are drawn independently at random. If the proposing side has preference lists consisting of only $O(1)$ hospitals (drawn independently at random using an arbitrary distribution), ~\citet{immorlica2005marriage} show that in such instances the number of candidates with more than one stable partner is a vanishingly small fraction, a phenomenon referred to as ``core convergence.''
They also consider partial lists on the hospital side and prove the following property (restated in~\citet{cai2022short}): if a hospital $h$ truncates its list at position $i$ (removing all residents of rank more than $i$), then $h$ has a stable partner of rank better than $i$ (in the original instance) if and only if $h$ is matched when running resident-proposing deferred-acceptance on the truncated instance. \citet{pittel1992likely} showed that resident preference lists of length $\Theta(\ln^2 n)$ were sufficient to obtain a perfect matching in uniform random matching markets.  This threshold was generalized to unbalanced random matching markets by~\citet{kanoria2021matching} and~\citet{potukuchi2025unbalanced}.  Partial preferences based on public scores was studied in~\citet{agarwal_et_al:LIPIcs.ICALP.2023.8}. 
All prior work on stable matching with partial lists assumes that the preferences follow a given distribution.  Our results generalize this line of work to any learnable preference distribution.

\paragraph{Algorithms with predictions.} Learned predictions have proven extremely effective in improving the algorithmic running time of many graph problems, such as matching~\citep{DinitzILMV21}, minimum cut~\citep{niaparast2025faster}, and flows~\citep{DaviesMVW,POLAK2024106487}, as well as many dynamic graph algorithms and data structures~\citep{McCauleyMNS23, mccauley2024incremental, McCauleyMoNi25, BrandFNP24,HenzingerSSY24,liu2023predicted}.  This paper is the first to leverage learned predictions for finding stable matchings efficiently.
\section{Preliminaries}
We denote the residents by $R = \{r_1,\ldots,r_n\}$ and the hospitals by $H = \{h_1,\ldots,h_n\}$, and refer to $R \sqcup H$ collectively as the set of \emph{agents}. Each agent $a$ has a strict preference list $L(a)$ over all agents of the opposite side. The resulting stable matching instance is $I = (R \sqcup H, L)$.

The rank of agent $a$ on preference list $\ell$ is denoted by $\text{rank}(\ell, a)$, where a higher-ranked agent has a smaller numerical rank. For integers $\rho \leq \rho'$, we use $\ell[\rho:\rho']$ to denote the inclusive sublist of $\ell$ from rank $\rho$ through $\rho'$. 
We use the notation $[n]$ to denote $\{1,2,\ldots, n\}$.

When preference lists \(L\) are not necessarily complete, we say that a matching \(\mu\) \emph{exists} in an instance \(I = (R \sqcup H, L)\) if for every agent \(a\), \(\mu(a) \in L(a)\). We say that \(\mu\) is \emph{eliminated} from \(I\) if there exists an agent \(a\) such that \(\mu(a) \notin L(a)\). 

We make frequent use of the following well-known properties of the DA algorithm~\citep{gusfield1989stable}. 

\smallskip
\begin{theorem}
For any instance of the stable matching problem, the DA algorithm terminates in $O(n^2)$ time, and the matching produced at termination is stable.
\end{theorem}

\smallskip
\begin{theorem}\label{thm:DA-optimality}
For a fixed proposing side, all executions of the deferred acceptance algorithm produce the same stable matching. This matching is optimal for the proposing side, in that each proposer receives its best attainable partner among all stable matchings, and pessimal for the receiving side, in that each receiver obtains its worst feasible partner across all stable matchings.
\end{theorem}
\section{Window-Truncated DA Algorithm}\label{sec:WDA}
This section studies the idea of truncating each preference list to a contiguous window that is likely to contain a stable match.

Let $\mu$ be any stable matching on a given instance $I$.  Each hospital $h_i$ receives a prediction~$\rho_i$ of the rank $\mu(h_i)$ on its preference list $L(h_i)$, along with an error bound $\eta_i$. 
Recall that for each hospital $h_i$, $\hat{L}(h_i) := L(h_i)[\rho_i-\eta_i:\rho_i+\eta_i]$ denotes the contiguous sublist containing residents within the error bound of its predicted match. 

\subsection{Instability in the Hospital-Proposing DA}

We first show that the natural approach of having hospitals propose only to agents within their predicted windows does not necessarily yield a stable matching, even if one exists within these windows.

\begin{lemma}\label{lem:hospital-proposing-is-unstable}
When hospitals propose only to residents within their predicted windows, the procedure does not always produce a stable assignment, even if one exists within those windows.
\end{lemma}
\begin{proof}
Consider the instance shown in Figure \ref{fig:counterexample}. It has a unique stable matching
\[\mu = \{(r_1,h_2), (r_2,h_4), (r_3,h_1), (r_4,h_3)\},\] 
drawn in black.  

\begin{figure}[H]
\begin{center}
    \begin{tikzpicture}[>= stealth, thick,
     node/.style = {circle, draw = black!70, fill = black!10, minimum size = 6mm, inner sep = 0pt, text width = 4mm, align = center},
     rnode/.style = {node}, hnode/.style = {node}]
    
    \node[rnode] (r1) {$r_1$};
    \node[rnode,below=0.25cm of r1] (r2) {$r_2$};
    \node[rnode,below=0.25cm of r2] (r3) {$r_3$};
    \node[rnode,below=0.25cm of r3] (r4) {$r_4$};
    
    \node[hnode,right=3cm of r1] (h1) {$h_1$};
    \node[hnode,below=0.25cm of h1] (h2) {$h_2$};
    \node[hnode,below=0.25cm of h2] (h3) {$h_3$};
    \node[hnode,below=0.25cm of h3] (h4) {$h_4$};
    
    \node[right=0.2cm of h1] {\small $[\grayhl{r_2},\grayhl{r_3},\grayhl{r_1},r_4]$};
    \node[right=0.2cm of h2] {\small $[r_3,\grayhl{r_2},\grayhl{r_1},\grayhl{r_4}]$};
    \node[right=0.2cm of h3] {\small $[\grayhl{r_4},\grayhl{r_1},{r_3},r_2]$};
    \node[right=0.2cm of h4] {\small $[\grayhl{r_3},\grayhl{r_2},\grayhl{r_4},r_1]$};
    
    \node[left=0.2cm of r1] {\small $[h_1,h_4,h_2,h_3]$};
    \node[left=0.2cm of r2] {\small $[h_4,h_3,h_1,h_2]$};
    \node[left=0.2cm of r3] {\small $[h_1,h_3,h_2,h_4]$};
    \node[left=0.2cm of r4] {\small $[h_1,h_4,h_2,h_3]$};
    
    \draw[black!70,ultra thick] (r1) -- (h2);
    \draw[black!70,ultra thick] (r2) -- (h4);
    \draw[black!70,ultra thick] (r3) -- (h1);
    \draw[black!70,ultra thick] (r4) -- (h3);
    
    \draw[gray!60] (r1) to [bend left = 20] (h2);
    \draw[gray!60] (r2) to [bend left = 20] (h1);
    \draw[gray!60] (r3) to [bend right = 20] (h4);
    \draw[gray!60] (r4) to [bend right = 20] (h3);
    \end{tikzpicture}
\end{center}
\caption{A counterexample showing that truncating preference lists of proposers can fail.}
\label{fig:counterexample}
\end{figure}
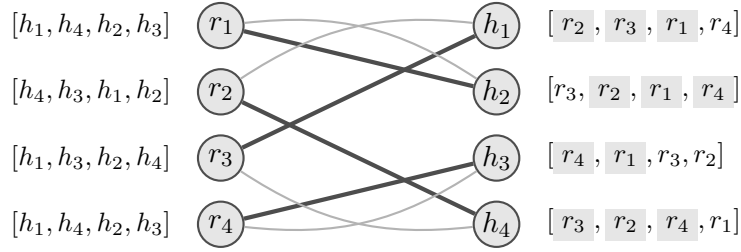

Suppose each hospital's predicted partner is exactly its match in $\mu$, and  each $\eta_i =1$. Each hospital’s pruned list is highlighted in gray in the figure.

If hospitals propose only to residents within their predicted windows, the algorithm outputs the matching
\[
\{(r_1,h_2), (r_2,h_1), (r_3,h_4), (r_4,h_3)\}.
\]  
This matching is not stable in the original instance because $(r_3,h_2)$ is a blocking pair.
\end{proof}

\subsection{Resident-Proposing Window-Truncated DA}\label{sec:wda_algorithm}
The above negative result motivates a more careful pruning scheme using resident-proposing DA. We call this the \textbf{window-truncated DA (WDA)} algorithm and prove that it always produces a stable matching whenever one exists within the pruned instance. 

Let $\hat{L}(h_i)$ denote the hospitals' pruned lists as before, i.e., 
$\hat{L}(h_i) := L(h_i)[\rho_i-\eta_i:\rho_i+\eta_i].$
For each resident $r_j$, let $\hat{L}(r_j)$ denote the set of hospitals $h_i \in L(r_j)$ for which $r_j \in \hat{L}(h_i)$, listed in the same order as in $L(r_j)$.

Now, run the DA algorithm on the new instance ${\hat{I} = (R \sqcup H, \hat{L})}$, with the residents proposing, to get a matching $\hat{\mu}$. 
We can immediately observe that the number of proposals of this algorithm is~$O(\sum_{i \in [n]} \eta_i)$.

We prove that WDA is correct: it will always give a stable matching if one exists in the predicted windows.
Suppose there exists a stable matching $\mu$ in $I$ satisfying $\mu(h_i) \in \hat{L}(h_i)$ for every $i \in [n]$. We begin by proving the following simple lemmas. 

\begin{lemma}\label{lem:WDA:mu-stable-in-hatI}
    The matching $\mu$ is stable in $\hat{I}$. 
\end{lemma} 
\begin{proof}
Since $\mu$ lies within $\hat{L}$, and the elements in $\hat{L}$ have the same relative order in $I$ and $\hat{I}$, any blocking pair with respect to $\mu$ in $\hat{I}$ is also a blocking pair in $I$. Therefore, since $\mu$ is stable in $I$, it is also stable in~$\hat{I}$. 
\end{proof}

\begin{lemma}\label{lem:WDA:residents-trade-up}
    For any resident $r$, 
   $\text{rank}(\hat{L}(r), \hat{\mu}(r)) \leq \text{rank}(\hat{L}(r), \mu(r)).$ 
\end{lemma}
\begin{proof}
The proof follows from Lemma \ref{lem:WDA:mu-stable-in-hatI} together with the fact that $\hat{\mu}$ is the resident-optimal stable matching in $\hat{I}$, as given by Theorem \ref{thm:DA-optimality}.
\end{proof}

We are now ready to show that WDA correctly finds a stable matching.

\begin{theorem}\label{thm:WDA}
Suppose there exists a stable matching $\mu$ in $I$ with $\mu(h_i)$ contained in the predicted window of $h_i$ for every hospital. Then the WDA algorithm produces a matching $\hat{\mu}$ that is stable in the original instance and matches all residents.
\end{theorem} 

\begin{proof}
Suppose $\hat{\mu}$ is not stable with respect to $I$.  Let $(r, h)$ be a blocking pair, that is, $r$ prefers $h$ to $\hat{\mu}(r)$ and $h$ prefers $r$ to $\hat{\mu}(h)$.  Using Lemma~\ref{lem:WDA:residents-trade-up},  we know that $r$ prefers $h$ to $\mu(r)$ as well, since residents only upgrade in $\hat{\mu}$.  An analog of Lemma~\ref{lem:WDA:residents-trade-up} for hospitals shows that hospitals only downgrade in $\hat{\mu}$, which is hospital-pessimal on $\hat{I}$.  Thus, $h$ prefers $\mu(h)$ to $\hat{\mu}(h)$.

There are two cases with respect to $r$'s position on $h$'s list in $I$.  
Suppose $h$ prefers $r$ to $\mu(h)$; then $(r, h)$ is also a blocking pair with respect to $\mu$ in $I$, which is a contradiction.  In the second case, $h$ prefers $r$ less than $\mu(h)$ but more than $\hat{\mu}(h)$.  This implies that $h \in \hat{L}(r)$, and that $(r,h)$ is a blocking pair with respect to $\hat{\mu}$ in $\hat{I}$, which is a contradiction.

Thus, $\hat{\mu}$ is stable with respect to $I$. 
Since all residents were matched in $\mu$, this implies that all residents are also matched in $\hat{\mu}$, since otherwise, the unmatched resident and hospital would form a blocking pair in $I$.
\end{proof}  

\paragraph{Discussion.} 
Notice that if $\mu$ is the resident-optimal stable matching in $I$, then $\hat{\mu} = \mu$. For an arbitrary stable matching $\mu$ in $I$, however, the algorithm may recover an intermediate stable matching, provided it lies within the pruned instance.
\section{Prefix-Truncated DA Algorithm}\label{sec:PDA}
In this section, we present a variant of the algorithm from Section \ref{sec:WDA} that eliminates the need for knowing the prediction errors a priori.   This algorithm prunes the preference lists of the hospitals after a certain rank~$\rho_i$ and retains the prefix of the original preferences up to $\rho_i$.
We call this variant the \textbf{prefix-truncated deferred-acceptance (PDA)} algorithm.

For each hospital $h_i$, let $\hat{L}(h_i) := L(h_i)[1:\rho_i]$.
As before, for each resident $r_j$, let $\hat{L}(r_j)$ consist of only those hospitals $h_i \in L(r_j)$ for which $r_j \in \hat{L}(h_i)$, in the same order as in $L(r_j)$.
We run the resident-proposing DA algorithm on the pruned instance ${\hat{I} = (R \sqcup H, \hat{L})}$ to obtain a matching $\hat{\mu}$. 

We show that if a perfect stable matching $\mu$ from the original instance $I$ exists within $\hat{I}$, then every hospital is matched in $\hat{\mu}$, and this matching is also stable in $I$. 

To show this, first we state the following observations.

\begin{observation}\label{lem:mu-stable-in-I'}
   If $\mu$ is a stable matching in $I$ that exists in $\hat{I}$, then $\mu$ is also stable in $\hat{I}$. 
\end{observation}
\begin{proof}
    For the sake of contradiction, suppose a blocking pair $(r,h)$ with respect to $\mu$ exists in $\hat{I}$. Then $r$ is ranked higher than $\mu(h)$ on $\hat{L}(h)$ and $h$ is ranked higher than $\mu(r)$ on $\hat{L}(r)$. 
    The elements in $\hat{L}$ have the same relative order in $I$ and $\hat{I}$. Hence, $(r,h)$ is also a blocking pair in $I$, which contradicts stability of~$\mu$ in $I$. 
\end{proof}

\begin{observation}\label{lem:stable-in-hatI-implies-stable-in-I}
    Every stable matching in $\hat{I}$ is also stable in $I$.
\end{observation}
\begin{proof}
    Consider a stable matching $\tilde{\mu}$ in $\hat{I}$, and, for sake of contradiction, suppose $(r,h)$ is a blocking pair with respect to $\tilde{\mu}$ in $I$. Since $\tilde{\mu}(h) \in \hat{L}(h)$, and $\hat{L}(h)$ is a prefix of $L(h)$, then ${r \in \hat{L}(h)}$ as well. Therefore, by definition, $h \in \hat{L}(r)$. Again, since the $\hat{L}$ sections of the preference lists have the same relative order in $I$ and $\hat{I}$, then $(r,h)$ is also a blocking pair with respect to $\tilde{\mu}$ in $\hat{I}$, which is a contradiction.  
\end{proof}

Finally,  we restate and use a well-known property of stable matchings, commonly referred to as ``the lone-wolf theorem'' or ``the rural hospital theorem.''
\begin{theorem}[Lone-Wolf Theorem~\citep{mcvitie1971stable}]\label{thm:lone-wolf}
Given a stable matching instance $I$, if an agent is unmatched in some stable matching $\mu$ over $I$, it must be unmatched in all stable matchings over $I$.
\end{theorem}

The following theorem can be derived from Observations~\ref{lem:mu-stable-in-I'} and~\ref{lem:stable-in-hatI-implies-stable-in-I}, together with Theorem \ref {thm:lone-wolf}.
\begin{theorem}\label{thm:underprediction-gives-imperfect-matching}
Let $\mu$ be any perfect stable matching with respect to $I$.  
If $\mu$ exists in $\hat{I}$, then $\hat{\mu}$ is perfect and stable with respect to $I$. 
\end{theorem}

We note that a similar claim when restricted to the case of individually truncated preference lists is observed in~\citet[Claim 3.1]{cai2022short}.

\subsection{Number of Proposals}
We bound the number of proposals made by the PDA algorithm, in the case where some stable matching exists within the pruned instance. 
Let $\mathcal{M}$ be the set of stable matchings in $
\hat{I}$. Note that by Observation \ref{lem:stable-in-hatI-implies-stable-in-I}, $\mathcal{M}$ is a subset of all stable matchings in $I$. Let $\tilde{\mu}$ be the stable matching in $\mathcal{M}$ that is \defn{closest} to the predicted ranks, defined as 
\[
\tilde{\mu} := \argmin_{M \in \mathcal{M}} \sum_{i\in[n]} \left[\rho_i - \text{rank}(\hat{L}(h_i), M(h_i))\right].
\]

\begin{lemma}\label{lem:algo-gives-closest-stable-matching}
The matching $\hat{\mu}$, found by the PDA algorithm, is the same as ${\tilde\mu}$.
\end{lemma}
\begin{proof}
   Observe that $\hat{\mu}$ is the hospital-pessimal (and resident-optimal) stable matching in $\hat{I}$ by Theorem~\ref{thm:DA-optimality}. Therefore, for every hospital $h_i$, among all partners attainable in $\mathcal{M}$, $\hat{\mu}(h_i)$ is the one closest to the end of the truncated list $\hat{L}(h_i)$. Hence, $\hat{\mu}$ is the stable matching that simultaneously minimizes the distance to the predicted rank on each hospital's list, and thus also minimizes the sum of these distances.
\end{proof}

Now, we bound the number of proposals made by the PDA algorithm by the distance to the nearest stable matching.

\begin{theorem}
\label{thm:pda_running_time}
    The number of proposals performed by the PDA algorithm is at most 
    \[\min_{M \in \mathcal{M}} \sum_{i \in [n]} \left[\rho_i - \text{rank}(\hat{L}(h_i), M(h_i)) + 1\right].\]
\end{theorem}
\begin{proof}
    The proof follows from Lemma \ref{lem:algo-gives-closest-stable-matching} and the observation that in the DA algorithm, each receiver’s tentative match weakly improves throughout the execution. In other words, each hospital $h_i$ only receives proposals from residents ranked between $\hat{\mu}(h_i)$ and $\rho_i$ on its preference list. 
\end{proof}

\subsection{Adapting to Incorrect Predictions}
In this section, we analyze the case where all stable matchings are eliminated from the pruned instance~$\hat{I}$.  We show that if this happens, then the unmatched hospitals are a subset of those that have been under-predicted, that is, all their stable partners were eliminated in the pruning.  This gives a certificate of which predictions were inaccurate. 

Let $\mu^*_H$ be the hospital-optimal stable matching in $I$. We say that a hospital $h$ is {\it under-predicted} if $\mu^*_H(h) \notin \hat{L}(h)$. 
The following theorem establishes a structural property of $\hat{\mu}$ in the case where pruning eliminates all stable matchings. 

\begin{theorem}\label{thm:unmatched-subseteq-underpredicted}
    Let $\hat{\mu}$ be the matching found by the PDA algorithm. If a hospital $h$ is unmatched in $\hat{\mu}$, then~$h$ must be under-predicted in $I$.
\end{theorem}
\begin{proof}
    For the sake of contradiction, suppose there is an unmatched hospital $h$, such that $\mu^*_H(h) \in \hat{L}(h)$. Without loss of generality, assume $h = h_1$ and $\mu^*_H(h) = r_1$. Since $h_1$ is unmatched, $r_1$ never proposed to $h_1$, despite $h_1$ being on its list. Thus, $r_1$ is matched to some other hospital, say $h_2$, which is ranked higher than $h_1$ on $\hat{L}(r_1)$. 
    
    Consequently, $h_2$ is not matched to $\mu^*_H(h_2) \equiv r_2$. 
    Since $(r_1,h_2)$ cannot be a blocking pair with respect to $\mu^*_H$, $r_2$ must be ranked higher than $r_1$ on $L(h_2)$, which means $h_2$ is not under-predicted. 
    Since $h_2$ is not under-predicted, implying $h_2 \in \hat{L}(r_2)$, and $r_2$ is not matched to $h_2$, then $r_2$ must be matched to some other hospital $h_3$ that it prefers to $h_2$. Again, this means $h_3$ is not matched to $\mu^*_H(h_3) \equiv r_3$, and since $(r_2,h_3)$ cannot be a blocking pair with respect to $\mu^*_H$, $r_3$ must be ranked higher than $r_2$ on $L(h_3)$, which implies that $h_3$ is not under-predicted. 
    Continuing inductively, we conclude that no hospital is under-predicted. Hence, $\mu^*_H$ exists in $\hat{I}$ and Theorem \ref{thm:underprediction-gives-imperfect-matching} implies that all hospitals are matched, a contradiction. 
\end{proof}

Thus, the PDA algorithm can iteratively extend the preference lists of only the unmatched hospitals and re-run resident-proposing DA until a perfect matching is found.
\section{Lower Bound}\label{sec:lower}

The analysis of the WDA and PDA algorithms in Sections~\ref{sec:WDA} and~\ref{sec:PDA} relies on the assumption that a stable
matching $\mu$ on the original instance $I$ survives in the pruned instance $\hat{I}$. It is natural to ask if this assumption is necessary.  In particular, can we improve upon the deferred-acceptance algorithm without such a promise?  In this section, we prove a lower bound showing this is not possible.  In particular, we prove that even with ``near-perfect'' predictions, without this promise, finding a stable matching still takes~$\Omega(n^2)$ time. 

We show that given a bound $\eta$, determining whether the maximum error is more or less than $\eta$, that is to say, determining if $\eta > \max_i \eta_i$ or $\eta \leq \max_i \eta_i$, requires $\Omega((n-2\eta)^2)$ queries to the agents' preference lists, and therefore $\Omega((n - 2\eta)^2)$ total time.  

One may further ask if it is possible to improve upon this bound if the prediction is near-perfect and contains only a small number of outliers---that is, can we obtain bounds using the average error, or if the number of agents with large error is small?  However, our lower bound holds even if we are promised that the prediction has only six agents whose predicted match is incorrect.  Even so, just finding those six agents requires $\Omega(n^2)$ time.

In summary, our lower bound shows that even near-perfect predictions are not sufficient to improve upon the worst-case cost of $\Omega(n^2)$.  This shows that the assumption that a stable matching exists in the pruned instance $\hat{I}$ is necessary in the analysis of our algorithms.

Concretely, in the remainder of this section we will prove the following theorem.
\begin{theorem}\label{thm:lower}
    Given a predicted matching $\hat{\mu}$ and an error bound $\eta$, any algorithm to determine if $\hat{\mu}$ has maximum error at most $\eta$ requires $\Omega((n - 2\eta)^2)$ queries to the agents' preference lists in the worst case, even if $\hat{\mu}$ can be made stable by changing the match of at most $6$ agents.
\end{theorem}

Our lower bound is an extension of~\cite{GonczarowskiNiOs14} and uses a reduction from the classic set disjointness problem. 
The {\bf set disjointness} problem, denoted $\disjoint(A,B)$, asks whether two subsets $A, B \subseteq [N]$ are disjoint. Formally, $\disjoint(A,B) = 1$ if $A\cap B = \emptyset$, and $\disjoint(A,B) = 0$ otherwise.  A classic result is that the randomized communication complexity of $\disjoint(A,B)$ is $\Omega(N)$, even if $|A\cap B|\leq 1$ ~\citep{KalyanasundaramSchintger92,Razborov90}.   

We embed two sets of size $(n - 2\eta)^2$ in a stable matching instance with a prediction $\hat{\mu}$, such that $\hat{\mu}$ has error $< \eta$ if and only if the sets are disjoint.  This generalizes~\citet[Lemma 4.1]{GonczarowskiNiOs14}, which embeds two sets in a stable matching instance to  show that verifying if a given matching is stable requires $\Omega(n^2)$ time.  Our proof generalizes theirs in two ways: first, we pad their embedding so that $\hat{\mu}$ is stable if and only if it has error at most $\eta$; second, we add two extra agents to ensure that the average error of $\hat{\mu}$ is small.

\paragraph{Constructing the stable matching instance.}
Consider an instance of the set disjointness problem in which the sets $A$ and $B$ are subsets of
\[U := \{(i,j): i,j \in [n - 2\eta-1], i \neq j\}.\]
We assume that $|A\cap B| \leq 1$.
    
We construct a corresponding stable matching instance with $n$ hospitals and $n$ residents.
We first partition the hospitals and residents into four groups:
\begin{itemize}[noitemsep]
    \item $S_h$ and $S_r$ are sets of $n - 2\eta - 1$ hospitals and residents respectively;  these encode the set disjointness instance
    \item $P_h^u$ and $P_r^u$ are sets of $\eta$ hospitals and residents respectively;  these serve as padding \emph{above} the predicted match
    \item $P_h^\ell$ and $P_r^\ell$ are sets of $\eta$ hospitals and residents respectively; these serve as padding \emph{below} the predicted match
    \item $B_h$ and $B_r$ contain a single hospital and resident; they serve as backup sets to ensure that, whenever ${A \cap B \neq \emptyset}$, only four hospitals need to change their match to reach a stable matching
\end{itemize}
To specify the preference lists, we use the notation $[X|Y]$ (where $X$ and $Y$ are sets) to denote the sequence consisting of the elements of $X$ ordered arbitrarily, followed by the elements of $Y$ ordered arbitrarily; we extend this notation to more than two sets in the natural way.

The preference lists of hospitals and residents in $P_h^u$ and $P_r^u$ are constructed as follows. For each $i$, the $i$th hospital in $P_h^u$ and $i$th resident in $P_r^u$ are first on each other's preference list;  the rest of their list is arbitrary.  These are essentially ``dummy agents''---hospital and resident $i$ in $P_h^u$ and $P_r^u$ will be matched in any stable matching.  The same construction is used for $P_h^\ell$ and $P_r^\ell$. 
    
Next, we describe the preference lists for the agents in $S_h$ and $S_r$.
The preference list of the $i$th hospital in $S_h$ is
\[
\left[\{ j \in S_r : (i,j) \in A \} \,|\, P_r^u \,|\, 
\{i\text{th resident in }S_r\} 
\,|\, 
B_r  \,|\, P_r^\ell \,|\, \{ j \in S_r : (i,j) \notin A \} \right].
\]
Similarly, the preference list of the $j$th resident in $S_r$ is
\[
\left[\{ i \in S_h : (i,j) \in B \} \,|\, P_h^u \,|\, \{j\text{th hospital in }S_h\} 
\,|\, 
B_h \,|\, P_h^\ell \,|\, \{ i \in S_h : (i,j) \notin B \}\right].
\]
Finally, the hospital in $B_h$ and the resident in $B_r$ have respective preference lists
\[
\left[S_r \,|\, P_r^u \,|\, P_r^\ell \,|\, B_r\right] \quad \text{and} \quad
\left[S_h \,|\, P_h^u \,|\, P_h^\ell \,|\, B_h\right].
\]

\paragraph{Predicted matching.}  We define the predicted matching $\hat{\mu}$ as follows:
\begin{itemize}[noitemsep]
    \item the $i$th hospital in $S_h$ is matched with the $i$th resident in $S_r$,
    \item the $i$th hospital in $P_h^u$ (resp. $P_h^\ell$) is matched with the $i$th resident in $P_r^u$ (resp. $P_r^\ell$), and
    \item the hospital in $B_h$ is matched with the resident in $B_r$.
\end{itemize}

We now prove that $\hat{\mu}$ has maximum error at most $\eta$, if, and only if, $A$ and $B$ are disjoint.  Moreover, we show that there are few agents with any error in $\hat{\mu}$, and in fact the \defn{average error} of $\hat{\mu}$ is at most 4.  These lemmas together prove Theorem~\ref{thm:lower}.

\begin{lemma}\label{lem:lowerbound-mapping}
    The sets $A$ and $B$ are disjoint if and only if the maximum error of $\hat{\mu}$ is at most $\eta$.
\end{lemma}
    
\begin{proof}
To begin, we show that if $A$ and $B$ are disjoint, then $\hat{\mu}$ is stable, which immediately implies that the maximum error is at most $\eta$, since it is 0. Suppose, for the sake of contradiction, that $\hat{\mu}$ contains an unstable pair. 
    
Note that no agent from $P_h^u$, $P_h^\ell$, $P_r^u$, or $P_r^\ell$ can be part of an unstable pair. Similarly, $B_h$ and $B_r$ cannot participate in an unstable pair: every hospital in $S_h$ is matched to a resident that it prefers over $B_r$, and the same holds for residents in $S_r$. Therefore, any unstable pair must consist of a hospital $h \in S_h$ and a resident $r \in S_r$.
    
Examining $h$'s preference list, the only residents that $h$ prefers over its current match are those corresponding to $(i,j) \in A$. Likewise, for $r$, the only hospitals it prefers over its match correspond to $(i,j) \in B$. This would contradict the assumption that $A$ and $B$ are disjoint.

Next, we show that if $\hat{\mu}$ has maximum error at most $\eta$, then $A$ and $B$ are disjoint. 
We do so by proving that if $\hat{\mu}$ is stable, then $A \cap B = \emptyset$, and if $\hat{\mu}$ is unstable, then it has maximum error more than $\eta$. 
  
First, suppose $\hat{\mu}$ is stable. If a pair $(\hat{i},\hat{j}) \in A \cap B$ exists, then the $\hat{i}$th hospital in $S_h$ and the $\hat{j}$th resident in $S_r$ would form a blocking pair with respect to $\hat{\mu}$. Therefore, stability of $\hat{\mu}$ implies that $A$ and $B$ must be disjoint.

Next, consider the case where $\hat{\mu}$ is not stable. We show that any stable matching $\mu$ must have distance more than $\eta$ from $\hat{\mu}$ on some preference list. Let $\mu$ be a stable matching. Since $\hat{\mu}$ is not stable, we have $\mu \neq \hat{\mu}$, which implies that there exists a pair $(h,r)\in \mu$ such that $(h,r) \notin \hat{\mu}$.    

The agents in $P_h^u$, $P_h^\ell$, $P_r^u$, and $P_r^\ell$ each have the same match in $\hat{\mu}$ and in all stable matchings. Hence, either $h \in S_h$ or $h \in B_h$. If $h \in B_h$, then $r \in S_r$, which implies that $\mu(h)$ is more than $\eta$ positions away from $\hat{\mu}(h)$ on $h$'s preference list. A similar argument holds if $r \in B_r$. If $h \in S_h$ and $r \in S_r$, since~$(h,r) \notin \hat{\mu}$, it follows that $h$ and $r$ are both more than $\eta$ positions away from each other’s matches in~$\hat{\mu}$.  
\end{proof}
  
\begin{lemma}\label{lem:averagerror}
    If $\hat{\mu}$ is not stable, then there exists a stable matching $\mu$ that  can be obtained by changing the match of at most $6$ agents in $\hat{\mu}$. Furthermore, $\hat{\mu}$ has average error at most $4$.
\end{lemma}
    
\begin{proof}
By the proof of Lemma \ref{lem:lowerbound-mapping}, if $\hat{\mu}$ is unstable, then $A$ and $B$ are not disjoint; recall that we assume $|A \cap B| \leq 1$. Let $(\hat{i},\hat{j}) \in A \cap B$ denote the unique overlapping element.

Construct the matching $\mu$ as follows. The $\hat{i}$th hospital in $S_h$ is matched to the $\hat{j}$th resident in $S_r$. The $\hat{i}$th resident in $S_r$ is matched to $B_h$, and the $\hat{j}$th hospital in $S_h$ is matched to $B_r$. All other hospitals and residents retain their matches from $\hat{\mu}$. 
     
We now show that $\mu$ is stable. As before, no hospital or resident in $P_h^u$, $P_h^\ell$, $P_r^u$, or $P_r^\ell$ can participate in an unstable pair.  First, consider the $\hat{i}$th hospital $h$ in $S_h$. If $h$ is part of an unstable pair, then there exists a resident that $h$ prefers over the $\hat{j}$th resident in $S_r$. Let this resident be the $j'$th resident in $S_r$. This resident must also prefer $h$ over its current match, which implies $(\hat{i},j') \in A \cap B$, contradicting our assumption that $(\hat{i},\hat{j})$ is the unique element in $A \cap B$.  

Next, consider any hospital $h \in S_h$ that is not the $\hat{i}$th hospital. Suppose $h$ is part of an unstable pair $(h,r)$. Let $h$ be the $i'$th hospital in $S_h$, and let $r$ be the $j'$th resident in $S_r$. Again, we would have $(i', j') \in A \cap B$, contradicting the uniqueness of $(\hat{i},\hat{j})$.  Notice that the resident in $B_r$ cannot be part of a blocking pair, since every hospital other than the $\hat{j}$th hospital in $S_h$ prefers its current match over this resident. 

Finally, consider $h \in B_h$, matched to $r \in S_r$. If $h$ is part of an unstable pair, then there exists some $r' \in S_r$ such that $(h,r')$ forms a blocking pair. But $r'$ is matched to a hospital in $S_h$ that it prefers over $h$, which contradicts the assumption of a blocking pair.  

Since all hospitals have been considered and none can participate in an unstable pair, we conclude that $\mu$ is stable.

We now bound the average error of the predicted matching $\hat{\mu}$ relative to $\mu$.  
The $\hat{i}$th hospital in $S_h$ and the $\hat{j}$th resident in $S_r$ each have error at most $n - \eta - 1$, since their lists have size $n$ and there are at least $\eta + 1$ entries below their match in $\hat{\mu}$.  
The $\hat{i}$th resident in $S_r$ and the $\hat{j}$th hospital in $S_h$ each have error $1$.  
Finally, $B_h$ and $B_r$ each have error at most $n$.  All other hospitals and residents have error $0$.  Thus, the total error is $4n - 2\eta \leq 4n$, and the average error is at most $4$.
\end{proof}
\section{Experimental Results}
In this section, we empirically evaluate our algorithms. In particular, we aim to address the following  questions to show that the theory is predictive of practice:
\begin{enumerate}

\item \textbf{Prediction concentration:} Can high-quality predictions be easily learned from past data?  Further, do the realized match ranks concentrate tightly enough that short prediction windows (as assumed in our analysis) are realistic in common market models?
\item \textbf{Efficiency scaling}: When prediction windows contain the true stable match, does the empirical number of proposals scale with the prediction error rather than market size, as predicted by our analysis?
\end{enumerate}

We conduct three sets of experiments, each corresponding to a different regime for generating preference lists~\footnote{The Python implementations  of the experiments are available at \href{https://github.com/helia-niaparast/stable-matching-with-predictions}{https://github.com/helia-niaparast/stable-matching-with-predictions}.}.

\subsection{Generating the Preferences}
\paragraph{Mallows markets.}

In our first set of experiments, preference lists are generated from the Mallows distribution \citep{mallows1957non}, a widely used probabilistic model for producing correlated rankings \citep{DBLP:conf/atal/BrilliantovaH22, hoffman2023stable, borodin2025natural}. 
In this model, agents’ preferences are centered around a common underlying ordering, with individual lists obtained by perturbing this central ranking with some amount of random noise. 
Such a structure reflects many real-world matching markets, where there is broad agreement about relative quality on both sides of the market, but not complete consensus. 

Formally, for each side of the market, we fix a reference ranking and generate each agent’s preference list independently according to a Mallows distribution around this reference. The amount of noise is controlled by a parameter $\phi \in [0,1]$, which we take to be the same on both sides of the market. When $\phi = 0$, all agents’ preference lists coincide exactly with the reference ranking, yielding perfectly aligned preferences. As $\phi$ increases, agents’ rankings deviate more from the reference, and when $\phi = 1$, preferences are completely random and independent. 

Under this model, the probability of generating a ranking $\pi$ is given by
\[\Pr(\pi) = \frac{1}{Z(\phi)} \, \phi^{d(\pi,\pi^*)},\]
where $\pi^*$ denotes the reference ranking, $d(\pi,\pi^*)$ is the Kendall-tau distance, i.e., the number of inverted pairs between $\pi$ and $\pi^*$, and $Z(\phi)$ is a normalization constant. Lower values of $\phi$ concentrate probability mass on rankings close to the reference, resulting in highly correlated preference lists, while higher values produce greater variability.

\paragraph{Real-world dating markets.}
In the second set of experiments, preference lists are generated from a distribution derived from data of an online dating platform \citep{brozovsky2007recommender}. We use the implementation of \citet{tziavelis2019equitable} to generate our lists.

Preferences are constructed by first assigning a coarse compatibility rating to every man--woman pair using a 
$10 \times 10$ probability table estimated from the dating-platform data. For each pair, two ratings are sampled from this table: one corresponding to the man’s evaluation of the woman and the other to the woman’s evaluation of the man. Each man then ranks women from highest to lowest according to his sampled ratings, and each woman ranks men analogously. Because the ratings take only ten levels, ties are frequent; these are broken by slightly favoring candidates who tend to receive higher ratings from many agents on the opposite side (a popularity adjustment), and then by adding a small random perturbation to ensure strict preference orderings. Finally, agents on each side are relabeled by popularity so that lower indices correspond to more popular agents, and the resulting preference lists are returned. 

\paragraph{Tiered markets.}
Our third set of experiments are run on tiered markets as described in \citet{ashlagi2020tiered}. Here, each side of the market is partitioned into a fixed number of tiers. Each tier is specified by two quantities: a fraction that determines what portion of agents belong to that tier, and a weight that determines how desirable agents in that tier are perceived to be. Preference lists are generated by a weighted sampling procedure: when an agent constructs its ranking, it repeatedly selects one remaining candidate at a time, with selection probabilities proportional to the candidates’ tier weights, removing the selected candidate after each draw until a complete strict ordering is obtained. 
In the experiments, both sides are divided into five tiers with fractions $(0.1, 0.2, 0.4, 0.2, 0.1)$ and weights $(50, 25, 10, 5, 1)$. The same tier structure applies to both sides, and agents retain their assigned tier throughout all experiments.

\subsection{Generating the Predictions}  
In each set of experiments, we fix the proposing side and the receiving side throughout. We repeat the experiment across a set of parameter settings: in the Mallows experiments, we fix $n = 500$ and vary the dispersion parameter $\phi$ and in the other two sets of experiments we vary the market size $n$.

For each parameter setting, we generate 50 random instances and compute the proposer-optimal stable matching in each. For every receiving agent, we record the rank of its match across these matchings. Using these 50 observations, we predict a lower and upper bound on the rank of the match for each receiving agent. The lower bound is set as the minimum observed rank minus three standard deviations, and the upper bound as the maximum observed rank plus three standard deviations.

\subsection{Experimental Setup}
For each parameter setting, after constructing the predictions, we generate 50 additional random instances from the same distribution. For each market, we run the (1) classic DA algorithm, (2) the WDA algorithm using the lower and upper predicted ranks, and (3) the PDA algorithm using only the upper predicted ranks, and record the number of proposals made by each algorithm.

For WDA, we additionally report the percentage of instances in which the output is stable. For PDA, if any receiving agents remain unmatched, we iteratively extend their preference lists and rerun the algorithm until a perfect matching is obtained (the initial extension has size $\lfloor n/8 \rfloor$ and is doubled in each iteration).

For both PDA and WDA, we also report the \defn{instance size}, defined as the sum of the lengths of the receiving agents’ preference lists after truncation. When lists are extended in PDA, we recompute the total length of all receiving agents’ lists and add this quantity to the previously recorded total.

\begin{figure*}[t]
  \centering
  \begin{subfigure}{0.32\textwidth}
    \centering
    \includegraphics[width=\linewidth]{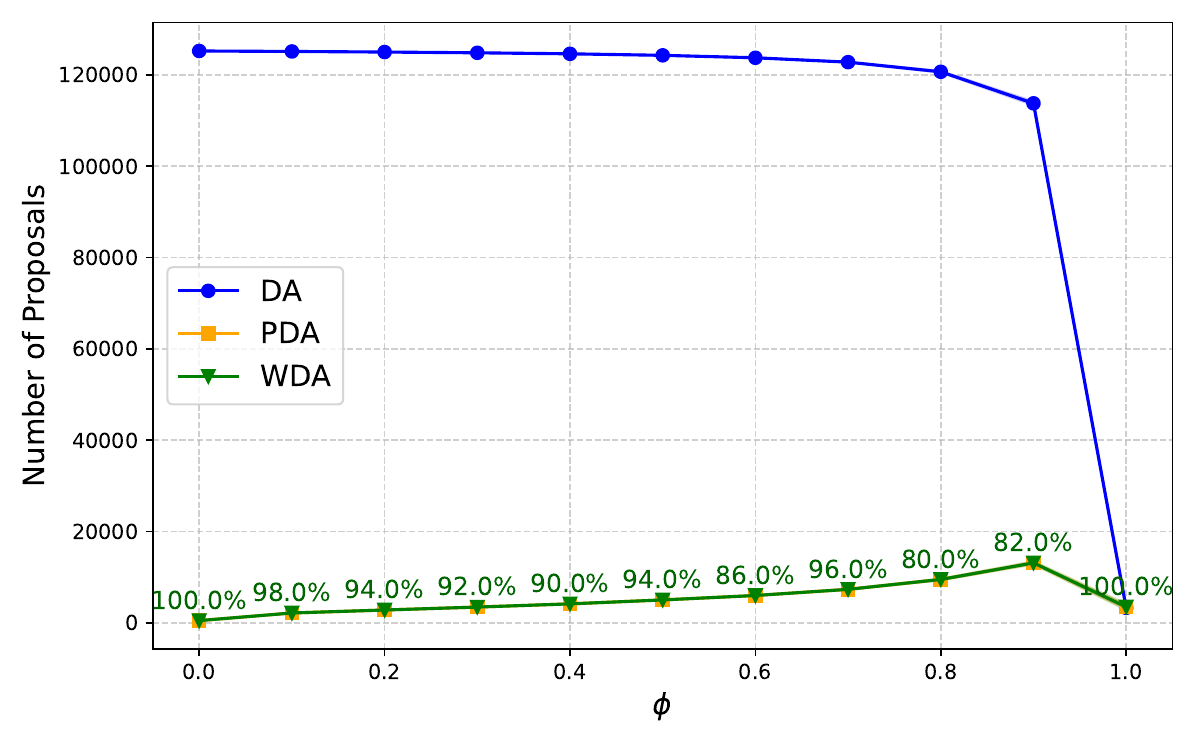}
    \caption{Mallows Markets}
    \label{fig:mallows}
  \end{subfigure}
  \hfill
  \begin{subfigure}{0.32\textwidth}
    \centering
    \includegraphics[width=\linewidth]{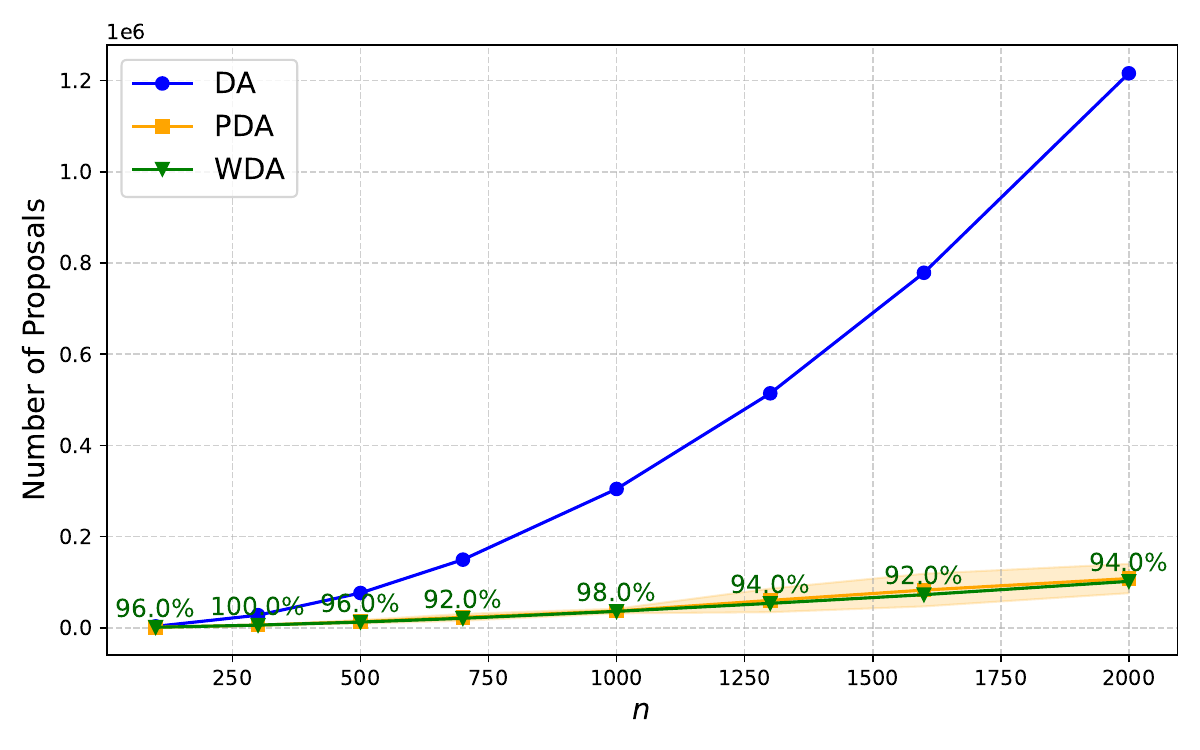}
    \caption{Dating Markets}
    \label{fig:dating}
  \end{subfigure}
  \hfill
  \begin{subfigure}{0.32\textwidth}
    \centering
    \includegraphics[width=\linewidth]{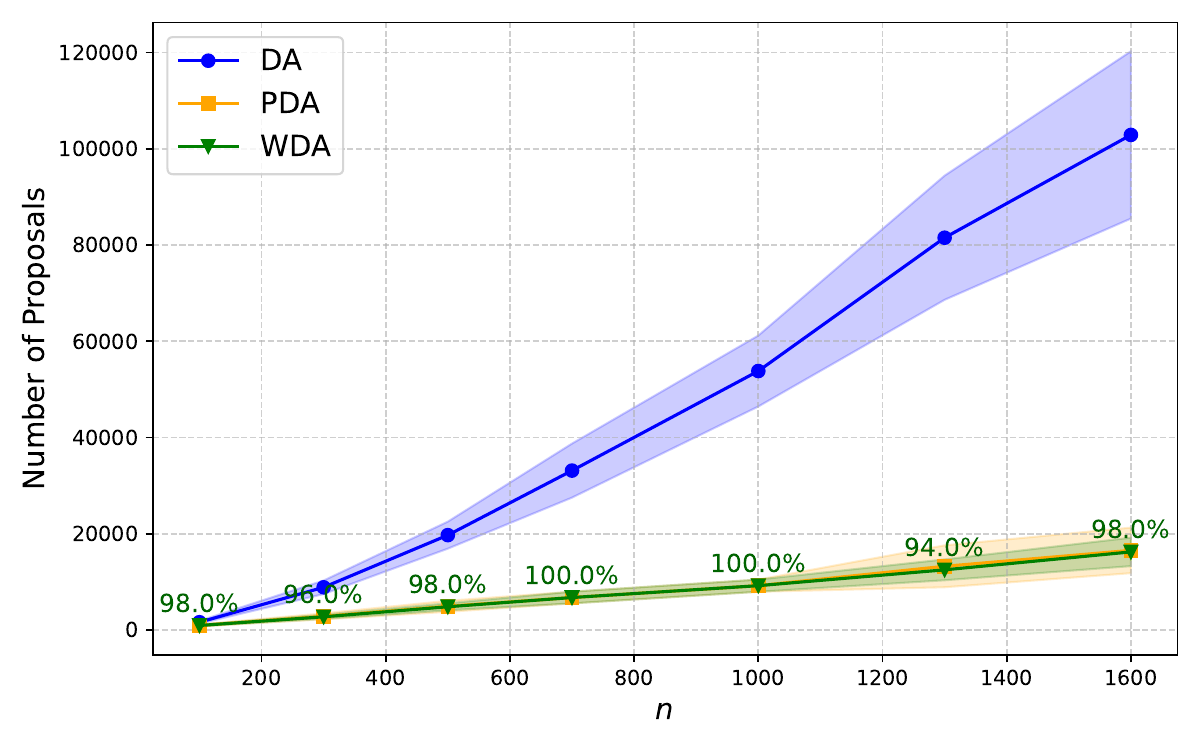}
    \caption{Tiered Markets}
    \label{fig:tiered}
  \end{subfigure}

  \caption{Proposal counts for classic DA, WDA, and PDA across the three market models. The green numbers show the percentage of instances in which WDA finds a stable matching. }
  \label{fig:prop-cnt}
\end{figure*}

\begin{figure*}[t]
  \centering
  \begin{subfigure}{0.32\textwidth}
    \centering
    \includegraphics[width=\linewidth]{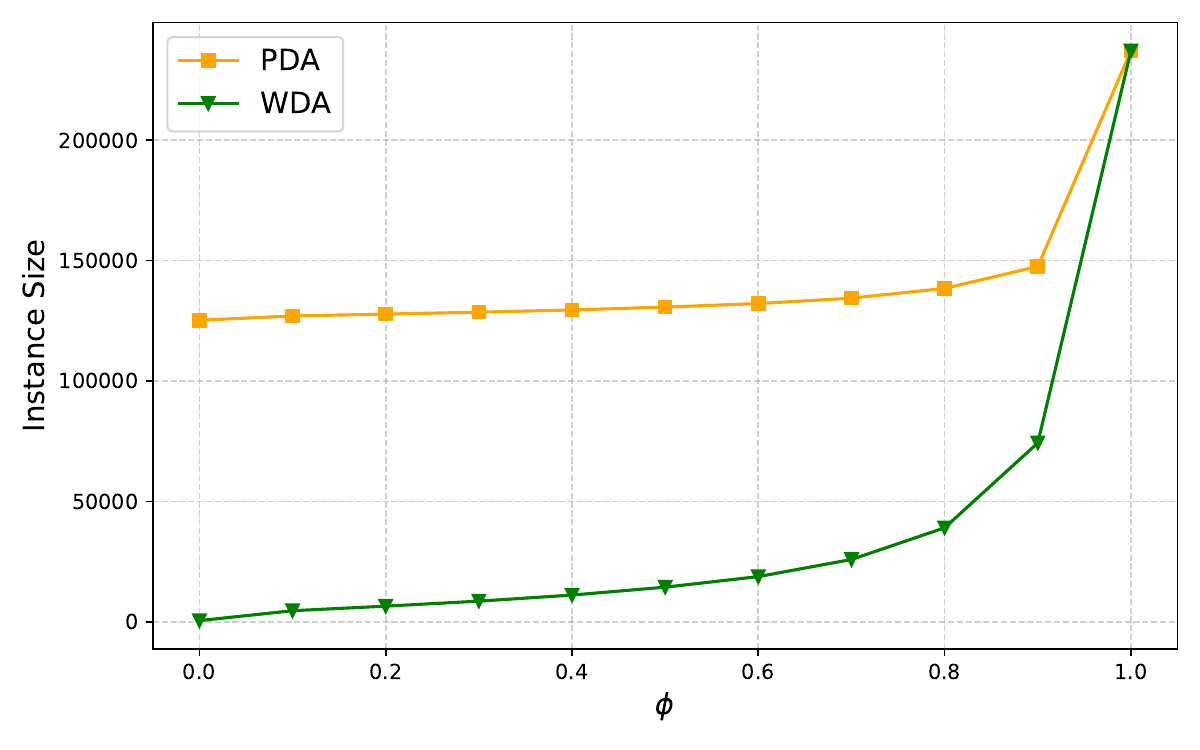}
    \caption{Mallows Markets}
    \label{fig:mallows2}
  \end{subfigure}
  \hfill
  \begin{subfigure}{0.32\textwidth}
    \centering
    \includegraphics[width=\linewidth]{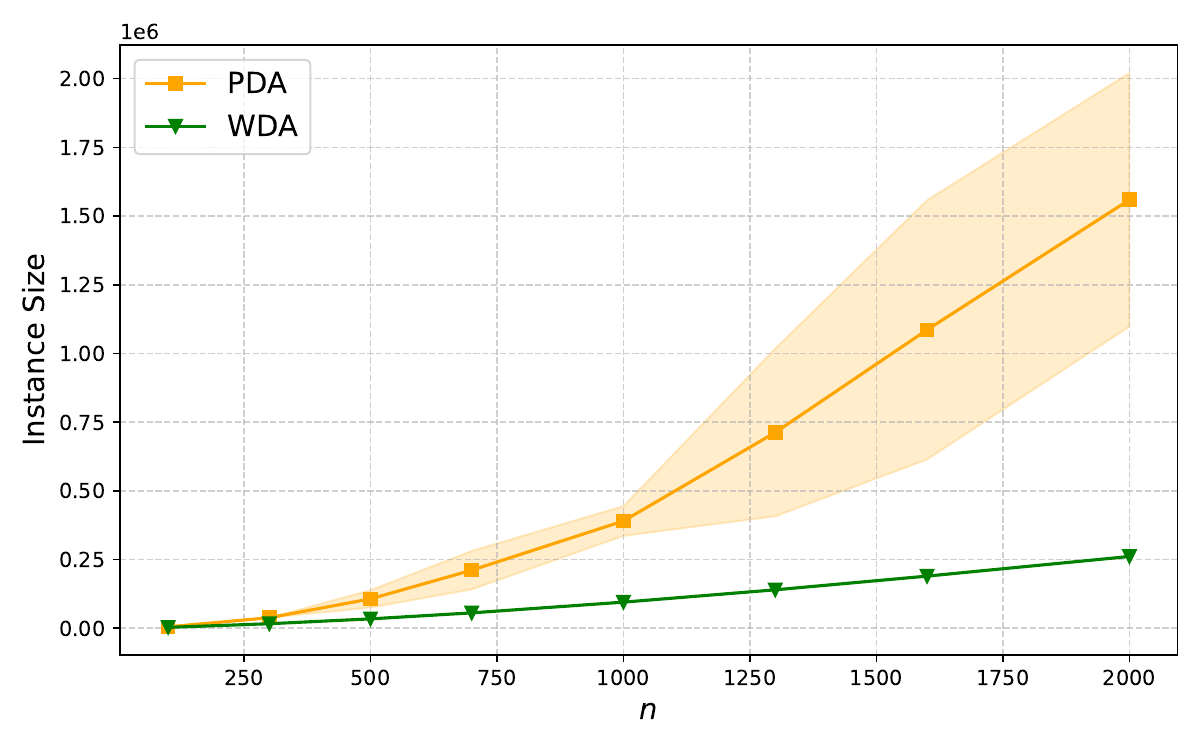}
    \caption{Dating Markets}
    \label{fig:dating2}
  \end{subfigure}
  \hfill
  \begin{subfigure}{0.32\textwidth}
    \centering
    \includegraphics[width=\linewidth]{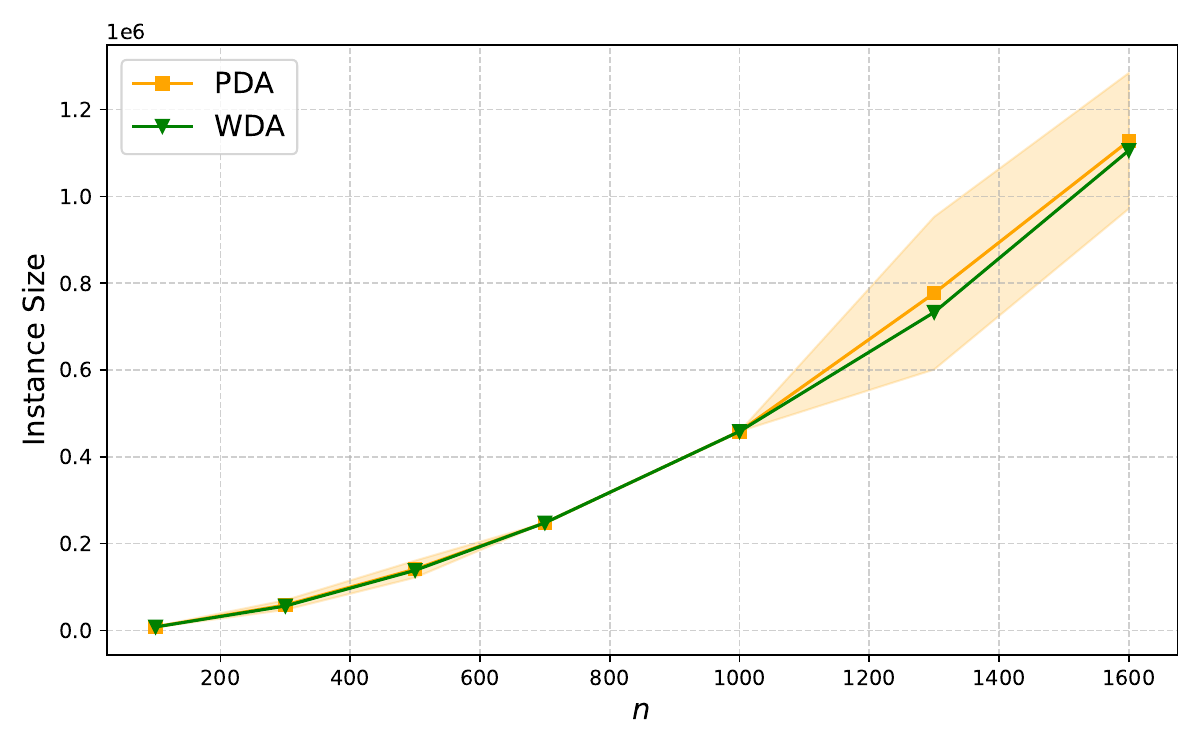}
    \caption{Tiered Markets}
    \label{fig:tiered2}
  \end{subfigure}

  \caption{Instance size for WDA and PDA across the three market models.}
  \label{fig:instance-size}
\end{figure*}

\subsection{Experimental Results}
Figures \ref{fig:prop-cnt} and \ref{fig:instance-size} show the results of the experiments. In all of the plots, the light shaded regions represent one standard deviation around the mean of the 50 observations; when they are not visible, the variance is very small. 

These results provide strong empirical validation for the theoretical bounds established in Sections~\ref{sec:WDA} and~\ref{sec:PDA}. Specifically, the high stability rates for WDA  confirm the robustness of Theorem \ref{thm:WDA}. Furthermore, the fact that PDA and WDA proposal counts track each other closely, and grow far more slowly than classic DA, suggests  that these market models exhibit the structural regularities
assumed by our analysis.

Overall, the empirical results closely match our theoretical bounds. When the predicted windows contain the true stable matches, both WDA and PDA achieve dramatic reductions in the number of proposals. Moreover, the output of WDA is mostly stable, and PDA remains efficient in practice even when the predicted prefixes occasionally do not contain any stable matching and the lists must be extended.

\paragraph{Mallows markets.} In the Mallows setting, as long as ${\phi \neq 1}$, meaning there is some correlation among the preferences, classic DA consistently makes substantially more proposals than WDA and PDA, often by about an order of magnitude or more. The case where $\phi = 1$ serves as a sanity check: when preferences are uniform random, classical theory predicts that DA already runs in near-linear time, and our prediction-guided algorithms cannot asymptotically improve upon it \citep{wilson1972analysis}.

\paragraph{Dating and tiered markets.} As we vary $n$, the number of proposals made by PDA and WDA grows at a far slower rate than in classic DA. As the market grows, the proposal counts under DA rise sharply, while the proposal counts for PDA and WDA increase much more gently, preserving a large and widening gap over the tested range.

In terms of stability, WDA produces a stable outcome in almost every instance. However, since WDA and PDA are consistently making roughly the same number of proposals, one might wonder what the advantage of WDA over PDA is, given that PDA is guaranteed to find a stable matching. The answer is in the comparison of instance size (Figure \ref{fig:instance-size}). In many real-world applications of stable matchings, constructing or submitting long preference lists is costly or constrained, so being able to use shorter preference lists, meaning that participants rank fewer options, is a significant advantage of WDA over PDA. 

Finally, it is worth noting that the proposal counts of PDA and WDA are often very close because they frequently output the same stable matching, and receiving agents are rarely proposed to by agents ranked above their lower match-rank prediction. Consequently, removing those parts of the lists typically does not affect the number of proposals.
\section{Conclusion}
In this paper, we develop a framework for analyzing stability when only subsets of preference lists are used. We study two different prediction-based truncation strategies.
When predictions on the hospitals’ true matches are accompanied by error bounds, we show that the resident-proposing window-truncated algorithm recovers the perfect stable matching and is substantially more efficient, while the hospital-proposing variant may return an unstable matching. 
We show that a more conservative prefix-truncated approach does not need to rely on known error bounds.
We prove a lower bound that shows that our assumption that at least one stable matching on the original instance exists in the pruned lists is necessary for efficiency, even if the predictions are near-perfect.

Together with our empirical evaluation, these results clarify when pruning is safe and how prediction-guided mechanisms can be applied reliably, providing a foundation for stable and scalable matching algorithms in large markets where full preference information is rarely feasible.

\bibliography{refs}
\bibliographystyle{plainnat}

\end{document}